%% file: main.tex
\documentclass[a4paper,conference]{IEEEtran}

\input{preamble}

\begin{document}
\title{A New Error Correction Scheme for \\ Physical Unclonable Functions}

\author{\IEEEauthorblockN{Sven Müelich\IEEEauthorrefmark{1},
Martin Bossert\IEEEauthorrefmark{1}}
\IEEEauthorblockA{\IEEEauthorrefmark{1}
              Ulm University, 
              Institute of Communications Engineering, 
              89081 Ulm, Germany\\
              Email: \{sven.mueelich, martin.bossert\}@uni-ulm.de}

}

\maketitle

\begin{abstract}
Error correction is an indispensable component when Physical Unclonable Functions (PUFs) are used in cryptographic applications. 
So far, there exist schemes that obtain helper data, which they need within the error correction process. 
We introduce a new scheme, which only uses an error correcting code without any further helper data.
The main idea is to construct for each PUF instance an individual code which contains the initial PUF response as codeword.
In this work we use LDPC codes, however other code classes are also possible.
Our scheme allows a trade-off between code rate and cryptographic security.
In addition, decoding with linear complexity is possible.
\end{abstract}

\begin{IEEEkeywords}
Physical Unclonable Functions, Secure Sketch, Helper Data Generation, Low-Density Parity-Check Codes, Cryptographic Key Generation and Storage
\end{IEEEkeywords}

\section{Motivation}
\emph{Physical Unclonable Functions} (PUFs) can be used for cryptographic purposes like identification, authentication, key generation and key storage. 
Using PUFs, keys do not have to be stored, but can be reproduced when needed.
However, usually some errors occur during the key reproduction process. 
Using error correction, the original key can be recovered. 
In order to perform error correction, schemes of so-called \emph{Secure Sketches} which use an error correcting code together with helper data were proposed in \cite{linnartz2003new} and \cite{dodis2004fuzzy}.
In this work, we suggest a new scheme, which only uses the code without any further helper data:
Section~\ref{sec:fundamentals} gives a more detailed description of PUFs and Low-Density Parity-Check (LDPC) codes. 
Section~\ref{sec:sketches} summarizes known Secure Sketch schemes.
In Section~\ref{sec:ldpc-scheme}, we propose our new scheme.
Section~\ref{sec:Examples} provides examples.
Finally, Section~\ref{sec:Conclusion} concludes the paper.

Notation: Let $\Cc(n,k,d)$ be a linear block code of length $n$, dimension $k$, and minimum distance $d$. 
Further, $\mathsf{wt}_H(\c)$ denotes the Hamming weight of vector $\c$.
The Hamming distance of two vectors $\c_1,\c_2$ is denoted by $\mathsf{d}_H(\c_1,\c_2)$.

\section{Fundamentals}
\label{sec:fundamentals}

\subsection{Physical Unclonable Functions}
A \emph{Physical Unclonable Function} (PUF) is a physical object which, according to an input (binary string, called challenge), produces an output (binary string, called response).
Since the calculation of the responses is based on randomness, which is intrinsic to the object  due to technical and physical limitations within the manufacturing process, devices which are identical in construction produce different, unique responses for the same inputs.
Various PUF constructions have been proposed. 
Most often, their randomness is either based on delays in electronic circuits or on the initialization behavior of memory cells. 
These devices are unclonable, since it is impossible to produce a device with a specific challenge-response behavior. 
Uniqueness and unclonability are properties which make PUFs suitable to use them for cryptographic applications. 
For example, a PUF response can be used as a cryptographic key. 
Since the randomness is static over the devices' lifetime, instead of storing the key in a non-volatile memory (and thereby making the system vulnerable to physical attacks), the response can be simply reproduced when the key is needed. 
However, there are two major problems: 
First, the responses of PUFs are not perfectly reproducible, since there is a variance caused by environmental factors (e.g. temperature, supply voltage, aging). 
Second, responses are not uniformly distributed. 
To tackle the first problem, error correcting codes can be applied. 
The non-uniformity can be solved by using cryptographic hash-functions. 
In this paper, we focus on the problem of nonperfect reproducible responses. 
For a comprehensive overview on PUFs, we refer to the literature \cite{boehm2012puf,maes2013puf,wachsmann2015puf}.

\subsection{LDPC Codes}
\emph{Low-Density Parity-Check (LDPC) codes} are widely used binary linear block codes, introduced by Gallager in \cite{gallager1962low}. 

\begin{definition}
\label{def:ldpc}
A $(\rho,\gamma)$-regular LDPC code of length $n$ and dimension $k$ is defined by a low-density \footnote{The density of a matrix is $\frac{\#\text{non-zero entries}}{\#\text{rows}\cdot \#\text{columns}}$} $(n-k) \times n$ parity check matrix $\H$ with the following properties:
\begin{itemize}
\item the number of ones in each row is $\rho \ll n$
\item the number of ones in each column is $\gamma \ll (n-k)$
\item for any two columns $c_j$, $c_{j'}$, there is at most one row $r_i$ such that $r_{ij} = r_{ij'} = 1$. The same holds for rows.\footnote{This property is not included in all definitions of LDPC codes. However, since the constructions based on finite geometries result in this property, we decided to add it to the definition.}
\end{itemize}
If the number of ones in each row or column is not constant, then the defined code is called irregular LDPC code.
\end{definition}
 
The number of ones in the parity check matrix of a $(\rho,\gamma)$-regular LDPC code is $\gamma\cdot n = \rho\cdot(n-k) $, the density is $\frac{\rho}{n} = \frac{\gamma}{n-k}$, and the coderate is $R \geq 1-\frac{\gamma}{\rho}$.

There exist many methods to construct LDPC codes.
Within this work we use constructions based on finite geometries and Reed-Solomon codes. 
For details on LDPC codes, we refer to standard textbooks on coding theory, e.g. \cite{bossert1999channel,kabatiansky2005error}.

\subsection{Finite geometries}
\label{subsec:geometry}
In \cite{kou2000low}, the construction of regular LDPC codes using finite geometries was proposed. We define Euclidean and projective geometries and explain how to derive LDPC codes.

\begin{definition}
\label{def:euclidean}
A Euclidean geometry $EG(m,q)$ consists of points $\{p_1,\dots,p_n\}$  where $p_i \in \Fqm$ ($n = q^m$), and lines $\{\Ll_1,\dots,\Ll_L\}$, such that the following axioms hold:
\begin{enumerate}
\item There is exactly one line between any points $p_i$ and $p_j$.
\item Two lines $\Ll_i$ and $\Ll_j$ either intersect in exactly one point or are parallel.
\item There are three points which are not on the same line. 
\item For each point $p_i \notin \Ll_j$, $\exists \Ll_l$ parallel to $\Ll_j$. 
\end{enumerate}
\end{definition}

As derived in \cite{bossert1999channel} and \cite{kabatiansky2005error}, $m$ and $q$ can be used to calculate
\begin{itemize}
	\item the number of points on each line
	\begin{align}
	\label{eqn:rhoeg}
	\rho = q,
	\end{align}
	\item the number of lines
	\begin{align}
	\label{eqn:Leg}
	L = \frac{q^{m-1}\cdot(q^m - 1)}{q-1},
	\end{align}
	\item and the number of lines which intersect in each point
	\begin{align}
	\label{eqn:gammaeg}
	\gamma = \frac{q^m -1}{q-1}.
	\end{align}
\end{itemize}

There are $\frac{q^m - 1}{q-1}$ different lines through the origin, namely
\begin{align}
	\Ll_j = \{\alpha a_j : \alpha \in \Fq, a_j \in \Fqm\setminus\{0\} \}.
\end{align}
For each line through the origin, there exist the parallel lines
\begin{align}
	\Ll_{i,j} =  \{a_i + \alpha a_j : \alpha \in \Fq, a_i \neq \alpha a_j\}.
\end{align}

We can derive an $L \times n$ matrix $\H$ from a geometry constructed using Definition~\ref{def:euclidean}, where each line of the geometry is a matrix row and each point is a column. 
The entries of that matrix are
\begin{align}
\label{eqn:paritycheck}
   h_{ij} =
   \begin{cases}
   1, & \text{if } p_j \in \Ll_i \\
   0, & \text{otherwise.}
   \end{cases}
\end{align}
In each row $\Ll_i$ there are $\rho$ ones which represent the points on line $\Ll_i$. 
In each column $p_j$ there are $\gamma$ ones which represent the lines intersecting in point $p_j$. 
Note that due to Definition~\ref{def:euclidean} Axiom~1, any two columns of $\H$ have at most one position where they both have a one as entry. 
Due to Definition~\ref{def:euclidean} Axiom~2, there is at most one position in which any two rows have entry one. 
We obtain a matrix of low density which fulfills all properties given in Definition~\ref{def:ldpc}. 
Hence, it can be used as parity check matrix of a ($\rho,\gamma$)-regular LDPC code.

We give an example to illustrate the construction. 
We choose $q=m=2$ and obtain the Euclidean geometry $EG(2,2)$ visualized in Figure~\ref{fig:eg} (left). Using $m$ and $q$ we calculate $\rho = 2$, $L = 6$, and $\gamma = 3$ according to Equations~(\ref{eqn:rhoeg})-(\ref{eqn:gammaeg}). Using this geometry, we apply Equation~(\ref{eqn:paritycheck}) to derive
\begin{align*}
\H = \begin{pmatrix}
1~1~0~0 \\
1~0~1~0 \\
1~0~0~1 \\
0~1~1~0 \\
0~1~0~1 \\
0~0~1~1 \\
\end{pmatrix},
\end{align*}
which fulfills all the properties given in Definition~\ref{def:ldpc} due to the construction of the geometry and hence can be used as parity check matrix of an LDPC code. 
To define an LDPC code of length $(n-k)$, the transposed of $\H$ can be used.

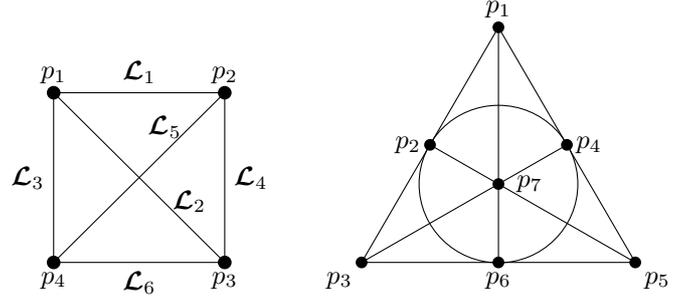
\begin{figure}
	\centering
\begin{tikzpicture}[scale=0.9]
\fill (-4.5,0) circle (0.1);
\fill (-4.5,2.5) circle (0.1);
\fill (-2,0) circle (0.1);
\fill (-2,2.5) circle (0.1);
\draw (-4.5,0) -- (-4.5, 2.5);
\draw (-4.5,0) -- (-2, 0);
\draw (-4.5,0) -- (-2, 2.5);
\draw (-2,0) -- (-2, 2.5);
\draw (-4.5,2.5) -- (-2, 2.5);
\draw (-2,0) -- (-4.5, 2.5);
\draw (-4.5,2.5) node [above] {$p_1$};
\draw (-2,2.5) node [above] {$p_2$};
\draw (-2,0) node [below] {$p_3$};
\draw (-4.5,0) node [below] {$p_4$};
\draw (-3.25,2.5) node [above] {$\Ll_1$};
\draw (-3.25,0) node [below] {$\Ll_6$};
\draw (-4.5,1.25) node [left] {$\Ll_3$};
\draw (-2,1.25) node [right] {$\Ll_4$};
\draw (-2.5,2) node [left] {$\Ll_5$};
\draw (-2.9,0.9) node [right] {$\Ll_2$};

\tikzset{mypoints/.style={fill=black,draw=black,thick}} 
\draw (0,0) coordinate[label=below left:$p_3$] (3) -- (60:4cm) coordinate[label=above:$p_1$] (1);
\draw (0,0) -- (4cm,0) coordinate[label=below right:$p_5$] (5);
\draw (5) -- (60:4cm);
\fill[mypoints] (3) circle (2pt);
\fill[mypoints] (1) circle (2pt);
\fill[mypoints] (5) circle (2pt);
\path (5) -- (120:4cm) coordinate (2out);
\coordinate[label=left:$p_2$] (2) at (intersection of 3--1 and 5--2out);
\fill[mypoints] (2) circle (2pt);
\path (0,0) -- (30:4cm) coordinate (4out);
\coordinate[label=right:$p_4$] (4) at (intersection of 5--1 and 3--4out);
\fill[mypoints] (4) circle (2pt);
\draw (3) -- (4);
\draw (5) -- (2);
\draw (1) -- (2.0cm,0) coordinate[label=below:$p_6$] (6);
\fill[mypoints] (2,0) circle (2pt);
\coordinate[label=right: $\,\,p_7$] (7) at (intersection of 3--4 and 2--5);
\fill[mypoints] (7) circle (2pt);
\draw (intersection of 3--4 and 2--5) circle (33pt);

\end{tikzpicture}
\caption{Euclidean geometry EG(2,2) constructed according to Definition~\ref{def:euclidean} (left), projective geometry PG(2,2) according to Definition~\ref{def:projective} (right).}
\label{fig:eg}
\end{figure}

Also projective geometry can be used. Each point in a projective geometry is described by a set of vectors instead of a field element. In contrast to Euclidean geometry, zero vector and parallel lines do not exist.

\begin{definition}
\label{def:projective}
A projective geometry $PG(m,q)$ consists of points ${p_1,\dots,p_n}$ $\big{(}n = \frac{q^{m+1}-1}{q-1}\big{)}$ and Lines $\Ll$. Each point is defined to be the set of vectors $\{\alpha(\beta_0,\beta_1,\beta_m) : \beta_j \in \Fq, \alpha \in \Fq \setminus \{0\}, (\beta_0,\beta_1,\beta_m) \neq (0,0,\dots,0)\}$.
Lines are a set of points which are defined according to the following axioms:
\begin{enumerate}
	\item There is exactly one line between any points $p_i$ and $p_j$.
	\item There are four points which are not on the same line.
	\item Two lines $\Ll_i$ and $\Ll_j$ intersect in exactly one point.
\end{enumerate}
\end{definition}

As derived in \cite{bossert1999channel}, $m$ and $q$ can be used in order to calculate
\begin{itemize}
	\item the number of points on each line
	\begin{align}
		\label{eqn:rhopg}
		\rho = q + 1,
	\end{align}
	\item the number of lines
	\begin{align}
		\label{eqn:Lpg}
		L = \frac{(q^{m+1}-1)(q^m -1)}{(q-1)^2 (q+1)},
	\end{align}
	\item and the number of lines which intersect in each point
	\begin{align}
		\label{eqn:gammapg}
		\gamma = \frac{q^m -1}{q-1}.
	\end{align}
\end{itemize}

A line through two points $a_i$ and $a_j$ is defined as
\begin{align}
	\Ll_{i,j} =  \{\alpha_i a_i + \alpha_j a_j : \alpha_i, \alpha_j \in \Fq, (\alpha_i,\alpha_j) \neq (0,0)\}.
\end{align}
Note that the points $a_i$ and $a_j$ are both a set of vectors according to Definition~\ref{def:projective}.
Figure~\ref{fig:eg} (right) visualizes PG(2,2).
The parity check matrix $\H$ of an LDPC code can be derived analog as in the Euclidean geometry case. 

\subsection{LDPC codes based on Reed-Solomon codes}
\label{subsec:rs}

A construction of LDPC codes based on Reed-Solomon (RS) codes was introduced in \cite{djurdjevic2003class}.
For a parameter $\rho$ ($1\leq \rho < q$) we construct a Reed-Solomon code $\Cc(q-1,q-\rho+1,\rho-1)$ over $\Fq$.
Shortening $\Cc$ by deleting the first $q-\rho-1$ information symbols results in code $\Cc_b(\rho,2,\rho-1)$. 
We choose a $\c \in \Cc_b$ with $\mathsf{wt}_H(\c)=\rho$ and generate the set $\Cc_b^{(1)} = \{\beta\cdot\c : \beta \in \Fq \}$. 
Based on this set, $\Cc_b$ can be partitioned into $q$ cosets $\Cc_b^{(1)},\dots,\Cc_b^{(q)}$.
For a vector $\c = (c_1,\dots,c_{\rho})$ we define the operation $z(\c) = (z(c_1),\dots,z(c_{\rho}))$, where each $z(c_j)$ is a sparse vector of length $q$ which has a one only at position $j$ when $c_j = \alpha_j$ for a generator $\alpha$ of $\Fq$. 
Using this operation we transform all sets $\Cc_b^{(i)}$  ($i = 1,\dots,q$) into $\Zz(\Cc_b^{(i)}) = \{z(\c): \c \in \Cc_b^{(i)}\}$.
Hence, each codeword is transformed to a vector of $\rho$ tuples, each having length $q$.
These vectors are used to construct matrices $\A_i \in \FF_2^{q \times \rho q}$ whose rows are the codewords of  $\Cc_b^{(i)}$ in transformed representation.
For $1 \leq \gamma \leq q$ we define 
\begin{align}
\H = 
\begin{pmatrix}
\A_1, \hdots,\A_\gamma
\end{pmatrix}^T.
\end{align}
$\H$ fulfills the properties of regular LDPC codes and hence can be used as parity check matrix.

\section{Known Secure Sketch Constructions}
\label{sec:sketches}

Secure Sketches were introduced in \cite{linnartz2003new} and \cite{dodis2004fuzzy} and address the problem that PUF responses are not perfectly reproducible. 
Figure~\ref{fig:sketch} visualizes a generic Secure Sketch which consists of an \emph{initialization phase} and a \emph{reproduction phase}.
During initialization, an initial response $\r_I$ is generated by the PUF (Figure~\ref{fig:sketch}~(a)). 
The \emph{Helper Data Generation} unit creates helper data (Figure~\ref{fig:sketch}~(b)), which are stored in the \emph{Helper Data Storage} for later usage in the reproduction phase.
This storage does not have to be protected, since the procedure is designed such that knowing the helper data does not reveal information about $\r_I$ that can be exploited by an adversary. 
Usually, initialization is performed in a secure environment.

In the reproduction phase, the initial response can be recovered, using a new generated response $\r$ and the helper data.
If the PUF produces $\r$ (Figure~\ref{fig:sketch}~(c)), a decoding algorithm in the \emph{Key Reproduction} unit is used, which reproduces $\r_I$ using $\r$ and the helper data, when $\mathsf{d}_H(\r_I,\r)$ lies within the error correction capability of the used code. 
Finally $\r_I$ is hashed to the actual key.

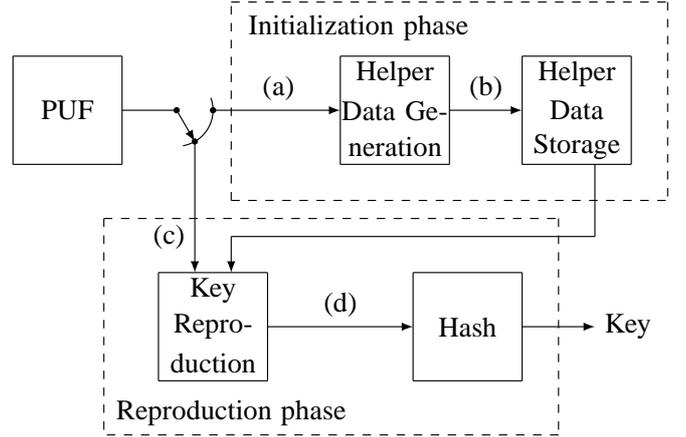
\begin{figure}[h]
	\input{securesketch.tex}
	\caption{Model of a Secure Sketch used for error correction for PUFs.}
	\label{fig:sketch}
\end{figure}

Two specific schemes, the \emph{Code-Offset Construction} and the {Syndrome Construction}, have been suggested in \cite{dodis2004fuzzy} and are used until now, e.g., in \cite{maes2012pufky}, \cite{muelich2014error}, and \cite{puchinger2015error}. We briefly explain these schemes in Sections~\ref{subsec:code-offset} and~\ref{subsec:syndrome}, before we introduce our new LDPC code based scheme in Section~\ref{sec:ldpc-scheme}.


\subsection{Code-Offset Construction}
\label{subsec:code-offset}
During the initialization phase, a randomly selected codeword $\c$ of a chosen code is added to the initial response $\r_I$ and the result $\h = \r_I \oplus \c$ is stored as helper data. 
After a new response $\r$ is generated in the reproduction phase,
\begin{eqnarray*}
	\c' &=& \r \oplus \h \\
	&=& \r \oplus (\r_I \oplus \c) \\
	&=& (\underbrace{\r \oplus \r_I}_{=\e}) \oplus \c 
\end{eqnarray*}
 is calculated using response $\r$ and helper data.
Since $\r_I$ and $\r$ have a small distance, $\r \oplus \r_I$ can be interpreted as error $\e$.
Hence, $\c'$ can be interpreted as received word and can be decoded. 
If $\mathsf{d}_H(\r_I,\r)$ is small enough, $\c$ and the decoded version of $\c'$ are the same and the initial response can be recovered by calculating $\r_I = \c \oplus \h$.

\subsection{Syndrome Construction}
\label{subsec:syndrome}
Using the syndrome construction, the helper data is the syndrome generated from initial response $\r_I$ and parity check matrix of a chosen code, i.e., $\mathbf{h} = \r_{I}^{T} \H$. 
If the PUF is evaluated in the reproduction phase, the calculation 
\begin{eqnarray*}
	\mathbf{Hr}^T &=& \mathbf{H}(\mathbf{r}_I^T \oplus \mathbf{e}^T) \\
	&=& \mathbf{Hr}_I^T \oplus \mathbf{He}^T \\
	&=& \mathbf{h}^T \oplus \mathbf{He}^T
\end{eqnarray*}
is performed.
After computing $\mathbf{Hr}^T \oplus \mathbf{h}^T$ the syndrome $\mathbf{He}^T$ remains and we can apply decoding.

\section{A new Secure Sketch Scheme}
\label{sec:ldpc-scheme}
In contrast to the schemes discussed in Section~\ref{sec:sketches}, our approach only needs the code and no further helper data.
The main idea of our construction is to design a code $\Cc$ for each PUF instance, such that the initial response $\r_I$, which we want to use as key, is a codeword. 

\subsection{Construction}
In the initialization phase, an individual code has to be generated for each PUF instance.
Let $\r_I$ be the initial response of length $n$. 
We choose one of the methods discussed in Sections~\ref{subsec:geometry} and \ref{subsec:rs} in order to generate a parity check matrix of an LDPC code of length~$n$ (Figure~\ref{fig:construction} (a)).
Note that there is no restriction on these LDPC code construction methods, others are also possible.
To enforce $\r_I$ to be codeword of the constructed code, we only keep rows which are dual to $\r_I$ (Figure~\ref{fig:construction} (b), bold rows), since the rows of the parity check matrix are codewords from the dual code. 
The other rows we ignore (Figure~\ref{fig:construction} (c)).
Therefore, we choose all lines $\Ll_i$ with $\langle\Ll_i,\r_I\rangle = 0$ to be rows of our parity check matrix~$\H$.
In order to get a memory efficient implementation, we do not need to generate the full codes before choosing the dual vectors. 
The non-dual vectors can be directly dropped during the code construction process.

However, due to the elimination of rows which are not dual to $\r_I$, the error correction capability of the constructed code becomes decreased in comparison to the full code. 
To increase the error correction capability again, we have to add more rows which are dual to $\r_I$ to get more parity check equations  (Figure~\ref{fig:construction} (d)). 
We can obtain these new equations using the remaining construction methods in order to generate more LDPC codes with the same parameters.

\begin{figure}[hbt]
	\centering
	\begin{tikzpicture}[scale=0.9]
	\draw (1.5,2.1) node {(a)};
	\draw (0,1.9) -- (3, 1.9); 
	\draw (0,1.8) -- (3, 1.8);
	\draw (0,1.7) -- (3, 1.7);
	\draw (0,1.6) -- (3, 1.6);
	\draw (0,1.5) -- (3, 1.5);
	\draw (0,1.4) -- (3, 1.4);
	\draw (0,1.3) -- (3, 1.3);
	\draw (0,1.2) -- (3, 1.2);
	\draw (0,1.1) -- (3, 1.1);
	\draw (0,1) -- (3, 1);
	\draw (0,0.9) -- (3, 0.9);
	\draw (0,0.8) -- (3, 0.8);
	\draw (0,0.7) -- (3, 0.7);
	\draw (0,0.6) -- (3, 0.6);
	\draw (0,0.5) -- (3, 0.5);
	\draw (0,0.4) -- (3, 0.4);
	\draw (0,0.3) -- (3, 0.3);
	\draw (0,0.2) -- (3, 0.2);
	\draw (0,0.1) -- (3, 0.1);
	\draw (0,0) -- (3, 0);
	
	\draw (5.5,2.1) node {(b)};
	\draw[line width=0.6mm] (4,1.9) -- (7, 1.9);  
	\draw (4,1.8) -- (7, 1.8);
	\draw (4,1.7) -- (7, 1.7);
	\draw[line width=0.6mm] (4,1.6) -- (7, 1.6);
	\draw[line width=0.6mm] (4,1.5) -- (7, 1.5);
	\draw (4,1.4) -- (7, 1.4);
	\draw (4,1.3) -- (7, 1.3);
	\draw (4,1.2) -- (7, 1.2);
	\draw[line width=0.6mm] (4,1.1) -- (7, 1.1);
	\draw(4,1) -- (7, 1);
	\draw[line width=0.6mm] (4,0.9) -- (7, 0.9);
	\draw (4,0.8) -- (7, 0.8);
	\draw (4,0.7) -- (7, 0.7);
	\draw[line width=0.6mm] (4,0.6) -- (7, 0.6);
	\draw (4,0.5) -- (7, 0.5);
	\draw[line width=0.6mm] (4,0.4) -- (7, 0.4);
	\draw (4,0.3) -- (7, 0.3);
	\draw[line width=0.6mm] (4,0.2) -- (7, 0.2);
	\draw[line width=0.6mm] (4,0.1) -- (7, 0.1);
	\draw (4,0) -- (7, 0);
	
	\draw (1.5,-0.8) node {(c)};
	\draw[line width=0.6mm] (0,-1.0) -- (3, -1.0);  
	\draw[line width=0.6mm] (0,-1.3) -- (3, -1.3);
	\draw[line width=0.6mm] (0,-1.4) -- (3, -1.4);
	\draw[line width=0.6mm] (0,-2.0) -- (3, -2.0);
	\draw[line width=0.6mm] (0,-2.1) -- (3, -2.1);
	\draw[line width=0.6mm] (0,-2.4) -- (3, -2.4);
	\draw[line width=0.6mm] (0,-2.6) -- (3, -2.6);
	\draw[line width=0.6mm] (0,-2.8) -- (3, -2.8);
	\draw[line width=0.6mm] (0,-2.9) -- (3, -2.9);

	\draw (5.5,-0.8) node {(d)};
	\draw[line width=0.6mm] (4,-1.0) -- (7, -1.0);  
	\draw[line width=0.6mm] (4,-1.2) -- (7, -1.2);
	\draw[line width=0.6mm] (4,-1.4) -- (7, -1.4);
	\draw[line width=0.6mm] (4,-1.6) -- (7, -1.6);
	\draw[line width=0.6mm] (4,-1.8) -- (7, -1.8);
	\draw[line width=0.6mm] (4,-2.0) -- (7, -2.0);
	\draw[line width=0.6mm] (4,-2.2) -- (7, -2.2);
	\draw[line width=0.6mm] (4,-2.4) -- (7, -2.4);
	\draw[line width=0.6mm] (4,-2.6) -- (7, -2.6);
	\draw[line width=0.6mm] (4,-2.6) -- (7, -2.6);
	\draw[line width=0.6mm] (4,-2.8) -- (7, -2.8);
	\draw[line width=0.6mm] (4,-3.0) -- (7, -3.0);
	\end{tikzpicture}
	\caption{Construction process of an LDPC code with a particular codeword. (a) Rows of a parity check matrix of a full LDPC code, obtained by using one of the construction methods given in Sections~\ref{subsec:geometry} and \ref{subsec:rs}. (b) Choose rows which are dual to the initial response $\r_I$. (c) Ignore all other rows. (d) Construct other rows which are dual to $\r_I$ in order to obtain good error correction properties.}
	\label{fig:construction}
\end{figure}
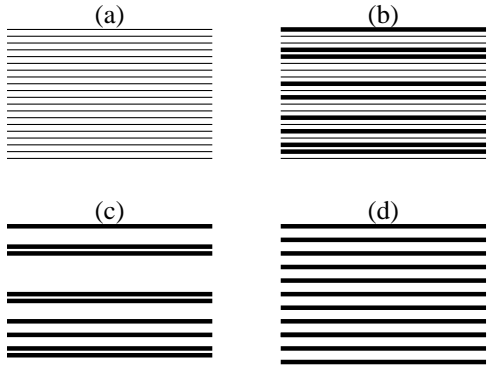

Adding new dual rows will increase the rank of the constructed matrix.
If we have reached the desired rank (remember that $k$ depends on the rank, since $\mathsf{rank}(\H) = n-k$) and still need more parity check equations to result in a good error correction performance, we add linear combinations of existing rows.
Doing this, we have to assure that the row weight of the new vectors does not become too large.

Since the combination of parity check equations from different code constructions might destroy the regularity of our constructed code, we create additional parity check equations whose entries increase the weight of low weight columns in order to guarantee that the column weight is roughly constant.
This property proved to be beneficial for the convergence behavior of our decoding algorithm (cf. Section~\ref{subsec:decoding}). 

When the PUF reproduces the key in the reproduction phase, the newly generated response $\r$ is interpreted as received word $\r = \r_i + \e$ and hence given to the decoding algorithm. 
The decoder outputs $\r_I$ when $\mathsf{d}_H(\r_I, \r)$ is small enough.

\subsection{Analysis}

We show that the constructed code is an LDPC code which has $\r_I$ as one of its codewords.
Let $\B_K = \{\b_i\}$ denote the rows of a parity check matrix obtained by a construction method $K \in \{EG,PG,RS\}$.
Let $\B = \{\b_j \in \B_K: \langle \b_j,\r_I\rangle = 0 \}$ be the set of rows which are dual to $\r_I$. Further, let L denote the cardinality of the set $\B$. 

\begin{lemma}
\label{lem:paritycheck}
$\H \subset \B$	is parity check matrix of an LDPC code $\Cc_L$, where $n-k = \mathrm{rank}(\B)$, with $\r_I \in \Cc_L$.
\end{lemma}	

\begin{proof}
$\B \subset \Cc_L^\perp$ is obvious since the rows of a parity check matrix are codewords of the dual code.
There exist $n-k$ rows $\b_i \in \B$ such that $\mathrm{rank}(\H) = n-k$.
Hence $\G$ is a $k \times n$ generator matrix of $\Cc_L$ and there exists an information word $\i$ such that $\i\cdot\G = \r_I$.
\end{proof}

Lemma~\ref{lem:paritycheck} can be directly applied to our construction where the rows of $\B$ are taken from different construction methods.

The correctness of our LDPC based Secure Sketch follows from the error correction capability of the constructed code. If $\mathsf{d}_H(\r_I, \r)$ is small enough, the code can correct $\r$ into $\r_I$ and hence reproduce the initial response.

Considering security, assume an adversary gains access to the parity check matrix $\H$. 
Since the correct response is a codeword, the uncertainty is still as large as the number of codewords.

Comparing the memory requirements of the different Secure Sketch constructions, our construction benefits from its sparse parity check matrix.
The helper data generated in the code-offset construction is a vector of length $n$, 
the syndrome construction only needs to store a vector of length $n-k$.
Additionally, both constructions need to store a representation of a code.
Since we do not have additional helper data in our scheme, we only need to store the code.
In contrast to previous constructions we use LDPC codes. They allow an efficient representation due to their sparse parity check matrices, which can be stored by using only one integer pair $(row, column)$ for each $1$.

\subsection{Decoding}
\label{subsec:decoding}

Several decoding algorithms which can be applied to LDPC codes have been introduced.
We use a simple, iterative bitflip procedure  according to \cite{bossert1986hard}, which flips a bit of the received word in each iteration, aiming for finally recovering the sent codeword.  
Let $\u_i$ be the $i$-th unit vector, which has a one at position $i$ and the rest zeros. 
$\mathrm{WT}(\r)$ denotes the Hamming weight of the syndrome of $\r$.
We implemented the algorithm as shown in Algorithm~\ref{alg:bitflip}.
\printalgoIEEE{
	\DontPrintSemicolon
	\KwIn{$\r = \c + \e$}
	\KwOut{decoding result $\hat{c}$.}
	Calculate $w = \mathrm{WT}(\r)$  \; \label{line:dec_1}
	\If{$\mathrm{wt}(w = 0)$} {
		$\hat{c} = r$ \;
		\Return{$\hat{c}$}}
	\For{$i = 1,\dots,n$} { 
		$\epsilon_i = \mathrm{WT}(\r + \u_i^T) + \Delta_i$ } \label{line:values}
	Find $j \in \{1,\dots,n\}$ with $\epsilon_j = \mathrm{min~} \epsilon_i$ \;
	$\r = \r + \u_j$ \;
	Goto 1
	\caption{Bitflip Decoding}
	\label{alg:bitflip}
}

Using multiple readouts, we obtain soft information which can be used by the bitflip algorithm in order to increase the error correction capability. 
Figure~\ref{fig:multiplereadouts} explains this process, using $m=3$ readouts:
We reproduce the PUF response $m$ times and hence obtain responses $\r_1,\dots,\r_m$.
The black circles in the figure highlight the error positions in the $m$ responses.
If an error occurs at position $i$ in some but not all of the responses, the response values differ at position $i$. 
By comparing the elements of the $m$ responses we can identify positions in which an error occurred.
In case all $m$ responses have the same value at position $i$, with very high probability there was no error at this position.  
For all positions $i = 1,\dots, n$ we define the soft information
\begin{align*}
	\Delta_i = 
	\begin{cases}
		\delta_1, & \text{if } \r_{1,i} == \r_{2,i} == \r_{3,i} \\
		\delta_2, & \text{otherwise}
	\end{cases}
\end{align*}
for $\delta_1>\delta_2 > 0$.
These values are used in Algorithm~\ref{alg:bitflip},  Line~\ref{line:values}, in order to increase the values of positions which should not be chosen for a bitflip. 
We have to decode all $m$ responses, so a trade-off between overhead due to readout and error correction capability can be found depending on application and used PUF construction.

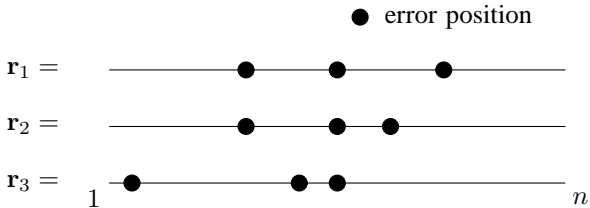
\begin{figure}[hbt]
	\centering
	\begin{tikzpicture}[scale=1.0]
	\draw (0,0) -- (6, 0); 
	\draw (0,0.75) -- (6, 0.75);
	\draw (0,1.5) -- (6, 1.5);
	\draw (-1,1.5) node {$\r_1 = $};
	\draw (-1,0.75) node {$\r_2 = $};
	\draw (-1,0) node {$\r_3 = $};
	\draw (-0.2,-0.2) node {$1$};
	\draw (6.2,-0.2) node {$n$};
	\filldraw (0.3,0) circle (3pt);
	\filldraw (2.5,0) circle (3pt);
	\filldraw (3,0) circle (3pt);
	\filldraw (1.8,0.75) circle (3pt);
	\filldraw (3,0.75) circle (3pt);
	\filldraw (3.7,0.75) circle (3pt);
	\filldraw (1.8,1.5) circle (3pt);
	\filldraw (3.0,1.5) circle (3pt);
	\filldraw (4.4,1.5) circle (3pt);
	
	\filldraw (3.3,2.2) circle (3pt);
	\draw (3.5,2.2) node [right] {error position};
	\end{tikzpicture}
	\caption{Obtaining soft information by using $m = 3$ readouts.}
	\label{fig:multiplereadouts}
\end{figure}

We analyze why decoding works.
Let $\r = \r_I + \e$ be a response extracted from a PUF.
The channel model usually used for PUFs is the \emph{Binary Symmetric Channel} (BSC) with a crossover probability $p$ which is given by the used PUF construction.
For all $\b_I \in \B$ it is
\begin{align*}
\langle \r, \b_i \rangle = \langle \r_I, \b_i \rangle + \langle \e, \b_i \rangle =  \langle \e, \b_i \rangle =: \s_i \in \{0,1\},
\end{align*}
If we calculate $\s_i$ for all $i = 1,\dots,L$, we get the syndrome vector $\s = (s_1,\dots,s_L)$.
It is $\mathrm{WT}(\r) = \mathrm{WT}(\e) = \mathrm{wt}_H(\s) = \sum_{j=1}^{L} = \langle \r, \b_j \rangle $.
Let 
\begin{align}
	\tilde{I}(j)
	\begin{cases}
	   	1, & \text{if } \mathrm{WT}(\r) >  \mathrm{WT}(\r + \u_j) \\
	   	0, & \text{if } \mathrm{WT}(\r) \leq  \mathrm{WT}(\r + \u_j)
	\end{cases}
\end{align}
be an indicator function, where a $1$ indicates an error at position $j = 1,\dots,n$.
We use the improved indicator function
\begin{align}
I(e) = \underset{j}{\mathrm{argmax}} \{\mathrm{WT}(\r) - \mathrm{WT}(\r + \u_j)\}.
\label{eqn:improvedind}
\end{align}
It is clear that
\begin{align*}
\mathrm{WT}(\r) \leq \mathrm{WT}(\r + \u_j) \Rightarrow \mathrm{WT}(\r) - \mathrm{WT}(\r + \u_j) \leq 0 \text{ and} \\
\mathrm{WT}(\r) > \mathrm{WT}(\r + \u_j) \Rightarrow \mathrm{WT}(\r) - \mathrm{WT}(\r + \u_j) > 0\text{.~~~~} 
\end{align*}
Function~\ref{eqn:improvedind} gives us error position $j$ of the smallest $\mathrm{WT}(\r + \u_j)$.
More details about the decoding algorithm can be found in \cite{bossert1986hard} and \cite[Chapter 7]{bossert1999channel}.


\section{Examples}
\label{sec:Examples}

We constructed codes of length 128 and 256 using $m = 3$ readouts.
During the construction process we can influence the rank of the obtained matrix and hence the dimension $k$ of the code. 
Note that the larger $k$, the larger the security level but the smaller the error correction capability.
Due to the flexibility of the code construction, the trade-off between security and error correction performance which best fits to application and used PUF construction can be found.
E.g., the error probabilities of different PUF constructions when reproducing responses given in \cite[Chapter 4.3.4]{maes2013puf} reveal details about the required error correction performance.
Figure~\ref{fig:n128k13} visualizes the block error probability of a constructed LDPC code of length 128 and dimension 13. 
Its parity check matrix has 881 rows.
Figure~\ref{fig:n128k56} shows the block error probability of a constructed LDPC code of same length, but dimension 56.
For this code, the parity check matrix has 349 rows.
The block error probability of a length 256 LDPC code of dimension 106 is visualized in Figure~\ref{fig:n256k106}. 
Its parity check matrix has 555 rows.
For decoding the length 128 (256) codes, we used $\delta_1=10$ and $\delta_2=6$ ($\delta_1=20$ and $\delta_2=12$). 
The block error probabilities of the LDPC codes are compared to BCH codes with similar parameters, since BCH is one of the code classes which are so far most often used in the PUF scenario \cite{maes2012pufky}.
Note that the length 128 LDPC codes outperform the corresponding BCH codes.

We calculated the block error probability for each crossover probability $p$ by
\begin{align}
	P_{Block}(p) = \sum_{i=0}^{t} P(i) \cdot P_{err}(i),
\end{align}
where $i$ is the number of errors, 
\begin{align}
	P(i) = {n \choose i}p^i(1-p)^{n-i}
\end{align}
is the probability of $i$ arbitrary errors in $n$ positions, and $P_{err}(i)$ is the relative amount of wrong decoded vectors $\e$ of weight $i$.
For the constructed LDPC codes, $P_{err}$ was determined via simulations.

\begin{figure}[tbh]
\centering
\includegraphics[width=\linewidth]{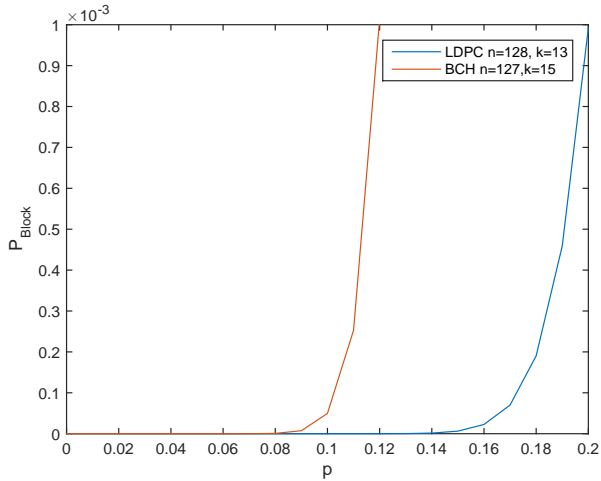}
\caption{Block error probability of a constructed LDPC code of length 128 and dimension 13 in comparison to a BCH code with similar parameters.}
\label{fig:n128k13}
\end{figure}

\begin{figure}[tbh]
	\centering
	\includegraphics[width=\linewidth]{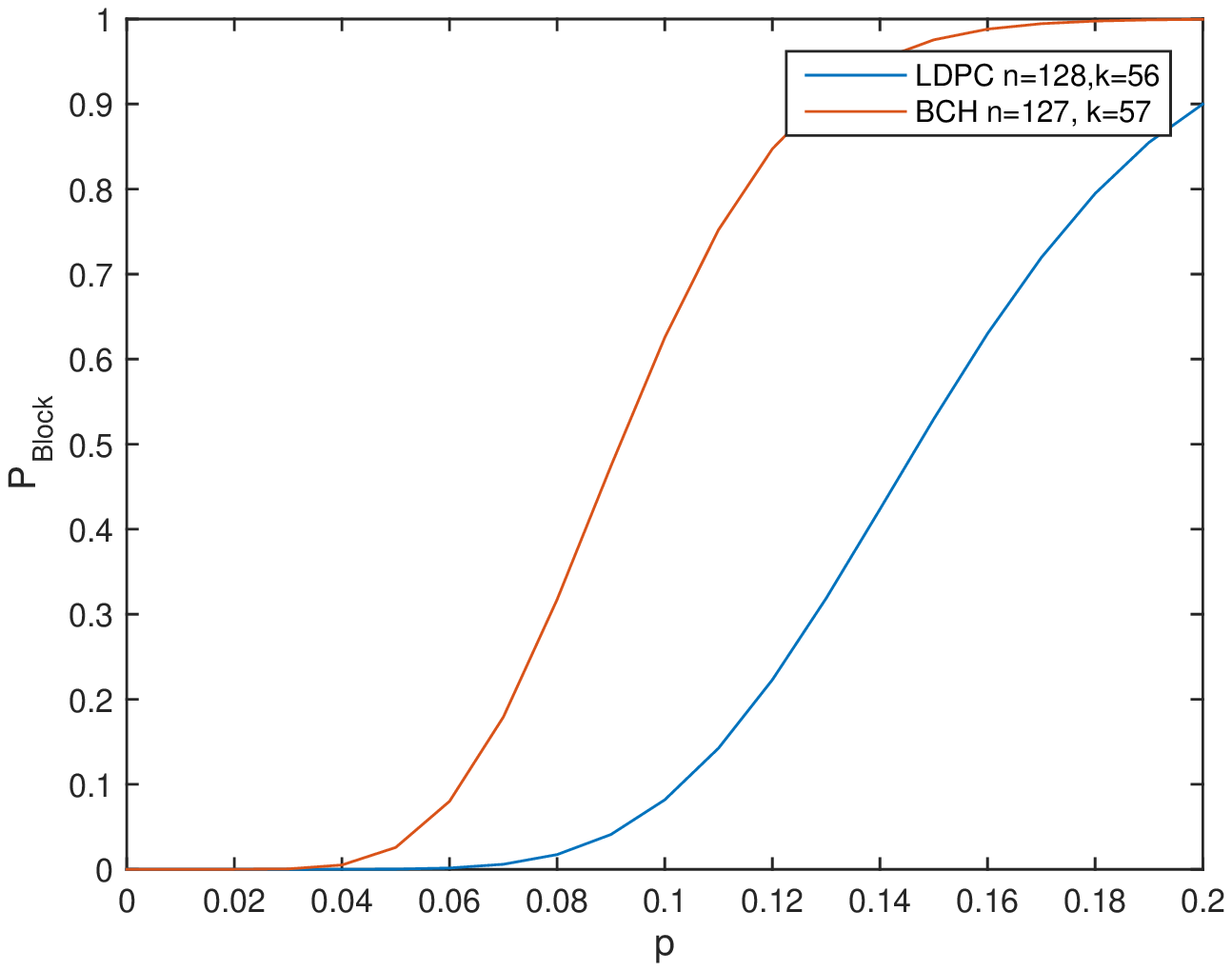}
	\caption{Block error probability of a constructed LDPC code of length 128 and dimension 56 in comparison to a BCH code with similar parameters.}
	\label{fig:n128k56}
\end{figure}

\begin{figure}[tbh]
	\centering
	\includegraphics[width=\linewidth]{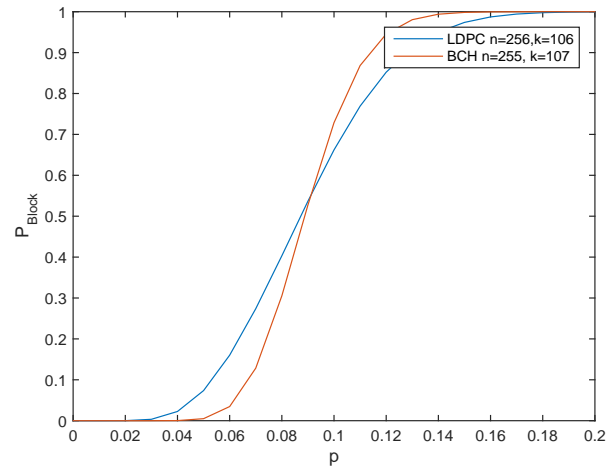}
	\caption{Block error probability of a constructed LDPC code of length 256 and dimension 106 in comparison to a BCH code with similar parameters.}
	\label{fig:n256k106}
\end{figure}
\section{Conclusions}
\label{sec:Conclusion}

We proposed a new Secure Sketch scheme which works by constructing an individual code around an initial PUF response. 
Since we do not need additional helper data, the scheme has a very plain structure.
LDPC codes allow a memory-efficient representation. The bitflip decoding algorithm enables decoding in linear time and can efficiently be implemented in hardware.

The examples presented in Section ~\ref{sec:Examples} show, that the construction results in codes which are suitable to be applied in the PUF scenario.
However, to apply the construction in practical scenarios, codes with larger dimensions have to be constructed to provide cryptographic security.
Also, further security issues have to be deeply investigated.
A further interesting question is, which other code classes can be constructed such that the initial PUF response is a codeword.

\section*{Acknowledgement}
The authors would like to thank Matthias Hiller for valuable discussions.

\bibliographystyle{IEEEtran}
\bibliography{main}


\end{document}

%% file: preamble.tex
\usepackage[utf8]{inputenc}
\usepackage[T1]{fontenc}
\usepackage{url}
\usepackage{ifthen}
\usepackage[cmex10]{amsmath} 
\usepackage{amsbsy}
\usepackage{diagbox}

\usepackage{amsmath}
\usepackage{amssymb}
\usepackage{amsthm}
\usepackage{mathrsfs}
\usepackage{graphicx}
\usepackage{color}
\usepackage{tikz}
\usepackage[colorlinks=true,citecolor=blue,linkcolor=blue]{hyperref}
\usepackage[linesnumbered,ruled,vlined,titlenumbered]{algorithm2e}

\newcommand{\bs}[1]{\boldsymbol{#1}}

\newcommand{\FF}{\mathbb{F}}
\newcommand{\Fq}{\mathbb{F}_q}
\newcommand{\Fqm}{\mathbb{F}_{q^{m}}}

\renewcommand{\H}{\mathbf{H}}
\newcommand{\G}{\mathbf{G}}
\newcommand{\B}{\mathbf{B}}
\newcommand{\A}{\mathbf{A}}

\renewcommand{\r}{\mathbf{r}}
\renewcommand{\c}{\mathbf{c}}
\newcommand{\h}{\mathbf{h}}
\newcommand{\e}{\mathbf{e}}
\renewcommand{\u}{\mathbf{u}}

\newcommand{\s}{\mathbf{s}}
\renewcommand{\i}{\mathbf{i}}
\renewcommand{\b}{\mathbf{b}}

\newcommand{\Cc}{\mathcal{C}}
\newcommand{\Zz}{\mathcal{Z}}
\newcommand{\Ll}{\bs{\mathcal{L}}}

\newtheorem{theorem}{Theorem}
\newtheorem{lemma}[theorem]{Lemma}

\newtheorem{definition}[theorem]{Definition}

\usepackage{varioref}        
\usepackage{fancyref}

\makeatletter
\def\mkfancyprefix#1#2{%
\expandafter\def\csname fancyref#1labelprefix\endcsname{#1}%
\begingroup\def\x{\endgroup\frefformat{plain}}%
    \expandafter\x\csname fancyref#1labelprefix\endcsname
    {\MakeLowercase{#2}\fancyrefdefaultspacing##1}%
\begingroup\def\x{\endgroup\Frefformat{plain}}%
    \expandafter\x\csname fancyref#1labelprefix\endcsname
    {#2\fancyrefdefaultspacing##1}%
\begingroup\def\x{\endgroup\frefformat{vario}}%
    \expandafter\x\csname fancyref#1labelprefix\endcsname
    {\MakeLowercase{#2}\fancyrefdefaultspacing##1##3}%
\begingroup\def\x{\endgroup\Frefformat{vario}}%
    \expandafter\x\csname fancyref#1labelprefix\endcsname
    {#2\fancyrefdefaultspacing##1##3}%
}
\makeatother
\fancyrefchangeprefix{\fancyrefeqlabelprefix}{eqn}
\mkfancyprefix{ssec}{Section}
\mkfancyprefix{tbl}{Table}
\mkfancyprefix{thm}{Theorem}
\mkfancyprefix{lem}{Lemma}
\mkfancyprefix{cor}{Corollary}
\mkfancyprefix{prop}{Proposition}
\mkfancyprefix{prob}{Problem}
\mkfancyprefix{alg}{Algorithm}
\mkfancyprefix{inv}{Invariant}
\mkfancyprefix{ex}{Example}
\mkfancyprefix{line}{Line}
\mkfancyprefix{def}{Definition}
\mkfancyprefix{itm}{Item}
\mkfancyprefix{app}{Appendix}
\mkfancyprefix{rem}{Remark}
\newcommand{\cref}[1]{\Fref{#1}}

\makeatletter
\newcommand{\removelatexerror}{\let\@latex@error\@gobble}
\makeatother

\newcommand{\printalgoIEEE}[1]
{{\centering
\scalebox{0.97}{
\removelatexerror
\begin{tabular}{p{\columnwidth}}
\begin{algorithm}[H]
 #1
\end{algorithm}
\end{tabular}
}
}
}

\renewcommand{\vec}[1]{\ensuremath{\mathbf{#1}}}

\renewcommand{\r}{\vec{r}}
\renewcommand{\c}{\vec{c}}
\renewcommand{\e}{\vec{e}}

%% file: securesketch.tex
{
\resizebox{0.49\textwidth}{!}{
\centering
\begin{tikzpicture}[scale=0.45]

\draw (0,0) -- (0,3);
\draw (0,3) -- (3,3);
\draw (3,3) -- (3,0);
\draw (3,0) -- (0,0);
\draw (1.5,1.5) node {PUF};

\draw (3,1.5) -- (4.5,1.5);
\draw[->,>=latex] (5,1.5-0.86602540378) -- (5,-3);
\draw[fill] (4.5,1.5) circle (2pt);
\draw[fill] (5.5,1.5) circle (2pt);
\draw[fill] (5,1.5-0.86602540378) circle (2pt);
\draw (5,-2) node [left] {(c)};
\draw [domain=-80:20] plot ({4.5+cos(\x)}, {1.5+sin(\x)});
\draw[->,>=latex] (4.5,1.5) -- (5,1.5-0.86602540378);

\draw (9,0) -- (9,3);
\draw (9,3) -- (12,3);
\draw (12,3) -- (12,0);
\draw (12,0) -- (9,0);
\draw (10.5,2.5) node {Helper};
\draw (10.5,1.5) node {Data Ge-};
\draw (10.5,0.5) node {neration};

\draw[->,>=latex] (5.5,1.5) -- (9,1.5);
\draw[->,>=latex] (12,1.5) -- (14,1.5);
\draw (7.3,1.5) node [above] {(a)};
\draw (13,1.5) node [above] {(b)};

\draw (14,0) -- (14,3);
\draw (14,3) -- (17,3);
\draw (17,3) -- (17,0);
\draw (17,0) -- (14,0);
\draw (15.5,2.5) node {Helper};
\draw (15.5,1.5) node {Data};
\draw (15.5,0.5) node {Storage};

\draw (16,0) -- (16,-2);
\draw (16,-2) -- (6,-2);
\draw[->,>=latex] (6,-2) -- (6,-3);

\draw (4,-3) -- (7,-3);
\draw (7,-3) -- (7,-6);
\draw (7,-6) -- (4,-6);
\draw (4,-6) -- (4,-3);
\draw (5.5,-3.5) node {Key};
\draw (5.5,-4.5) node {Repro-};
\draw (5.5,-5.5) node {duction};

\draw[->,>=latex] (7,-4.5) -- (11,-4.5);
\draw (9,-4.5) node [above] {(d)};

\draw (11,-3) -- (14,-3);
\draw (14,-3) -- (14,-6);
\draw (14,-6) -- (11,-6);
\draw (11,-6) -- (11,-3);
\draw (12.5,-4.5) node {Hash};

\draw[->,>=latex] (14,-4.5) -- (16,-4.5);
\draw (16,-4.5) node [right] {Key};

\draw [dashed] (6,-1) -- (6,4.5);
\draw [dashed] (6,4.5) -- (18,4.5);
\draw [dashed] (18,4.5) -- (18,-1);
\draw [dashed] (6,-1) -- (18,-1);
\draw (9.5,4.4) node [below] {Initialization phase};

\draw [dashed] (2.5,-1.5) -- (2.5,-7.5);
\draw [dashed] (2.5,-1.5) -- (15,-1.5);
\draw [dashed] (15,-1.5) -- (15,-7.5);
\draw [dashed] (15,-7.5) -- (2.5,-7.5);
\draw (6,-7.5) node [above] {Reproduction phase};

\end{tikzpicture}
}
}